\documentclass[conference]{IEEEtran}
\usepackage{cite}
\usepackage{array}

\usepackage{hyperref}
\usepackage{graphicx,footnote}
\usepackage{color}
\usepackage[cmex10]{amsmath}
\usepackage{amssymb,amsfonts}
\usepackage{amsthm}

\newtheorem{theorem}{Theorem}
\newtheorem{lemma}{Lemma}

\usepackage{url}

\begin{document}
%
\title{
Local Popularity Based Collaborative Filters}

\author{
\IEEEauthorblockN{Kishor Barman}
\IEEEauthorblockA{School of Technology and Computer Science\\
Tata Institute of Fundamental Research \\
Mumbai, India\\
Email: kishor@tcs.tifr.res.in}
\and
\IEEEauthorblockN{Onkar Dabeer}
\IEEEauthorblockA{School of Technology and Computer Science\\
Tata Institute of Fundamental Research\\
Mumbai, India\\
Email: onkar@tcs.tifr.res.in}
}

\maketitle

\begin{abstract}
Motivated by applications such as recommendation
systems, we consider the estimation of a binary random field $\mathbf X$
obtained by {\it unknown} row and column permutations of a block constant
random matrix. The estimation of $\mathbf X$ is based on observations
$\mathbf Y$, which are obtained by passing entries of $\mathbf X$ through a binary
symmetric channel (BSC) (representing noisy user behavior) and an
erasure channel (representing missing data).
We analyze an estimation algorithm based on local popularity. We
study the bit error rate (BER) in the limit as the matrix size
approaches infinity and the erasure rate approaches unity at a
specified rate. Our main result identifies three regimes characterized
by the cluster size and erasure rate. In one regime, the algorithm has
asymptotically zero BER, in another regime the BER is bounded away from
0 and 1/2, while in the remaining regime, the algorithm fails and BER
approaches 1/2. Numerical results for the Movielens dataset and
comparison with earlier work is also given.
\end{abstract}

\section{Introduction}


Recommendation systems are commonly used in e-commerce to suggest relevant content to users. One approach considers the user-item rating matrix, predicts the missing entries, and recommends items based on the predicted values (for example, see \cite{netflixprize}). Recently, a number of researchers have considered mathematical models for this problem and studied fundamental limits. One model assumes the rating matrix to be a low-rank random matrix (\cite{Candes1,Montanari1,Bresler1}), and then  bounds on the number of samples needed to recover the {\it complete} matrix with high probability are obtained. In another model (\cite{Aditya1, Aditya2}), the rating matrix is assumed to be obtained from a block constant matrix by applying unknown row and column permutations, a noisy discrete memoryless channel representing noisy user behavior, and an erasure channel denoting missing entries. The goal for such a model is not matrix completion, but estimation of the underlying ``noiseless'' matrix. In  \cite{Aditya1, Aditya2},  the probability of error in recovering the {\it entire} matrix  for fixed erasure rate is considered, and threshold results reminiscent of the channel coding theorem (but with different scaling) are established.

 In this paper, we consider the model in \cite{Aditya2}, but we allow the erasure rate to approach unity, and focus on the BER - the probability of error that a specific recommendation fails. We analyze the BER for a specific algorithm, which makes recommendations based on ``local popularity''. Such an analysis is of interest for two reasons:
\begin{itemize}
\item It gives an upper bound on achievable BER;
\item The local popularity algorithm used is motivated by algorithms used in practice \cite{Linden1}, and has lower complexity compared to those in the above mentioned references.
\item The algorithm has competitive empirical performance on real datasets
such as the Movielens data \cite{MovieLens}. For example, next we compare the
algorithm with OptSpace \cite{Montanari1} on Movielens data. While OptSpace uses
ratings on the scale 1-5 given by Movielens, in our algorithm we
quantize the ratings as follows: 4,5 are mapped to 1, while 1-3 are
mapped to 0. (Similarly, the output of OptSpace is quantized to
$\{0,1\}$.) We find that the local algorithm yields a BER of 0.091,
while on the same test data, OptSpace gives a BER of 0.107. Thus the
performance of both algorithms is similar. (More detailed simulation
results will be presented in a future publication.)
\end{itemize}


In this paper, we seek to understand the reason for the competitive
performance of the relatively simple local
algorithm by analyzing its BER for the model proposed in \cite{Aditya1}.  Suppose that the matrix is of size $n \times n$ and the erasure probability $\epsilon = 1-c/n^{\alpha}$. If  $\alpha \in [0,1/2)$, then our main result says that if the cluster size is greater than $n^{\alpha -\gamma_n}$ where $\gamma_n \rightarrow 0$, then the BER approaches 0, but if the cluster size is less than $n^{\alpha -\gamma}$, $\gamma >0$, the BER is bounded away from zero and a lower bound is  obtained in terms of the observation noise and $\gamma$. For $\alpha >1/2$, BER always approaches 1/2. Due to space constraints we only
provide an outline of the proofs; the details with additional results
will be reported in a journal submission.

The rest of the paper is organized as follows. In Section \ref{sec:basic_setup}, we describe our model, the local popularity algorithm, and establish notation. The main results are stated and discussed in Section \ref{sec:main}. The proof of the main result is given in \ref{sec:proof_main} and some related lemmas are established in \ref{sec:proof_lemma}. The conclusion of given in Section \ref{sec:con}.

\section{Basic Setup}
\label{sec:basic_setup}
In Section \ref{subsec:model} we describe our model, and discuss a local popularity based algorithm in Section \ref{subsec:local}.
\subsection{The Model}
\label{subsec:model}

We consider an $n\times n$ rating matrix $\mathbf X$ whose entries are binary. The rows of the matrix represent users and the columns represent items. Suppose $\mathcal A=\{A_i\}_{i=1}^r$ and $\mathcal B=\{B_i\}_{i=1}^r$ are row and column partitions respectively, representing sets of similar users and items. We assume that for all $i=1,\dots,r$ we have $|A_i|=|B_i|=k$. The sets $A_i\times B_j$ are the clusters of the matrix and they are unknown. If $(p,q)\in A_i\times B_j$, then $\mathbf X(p,q)=\xi_{ij}$ where $\xi_{ij}$ are i.i.d. Bernoulli(1/2). This matrix $\mathbf X$ is passed through a memoryless binary  symmetric channel (BSC) with parameter $p$, and then through an erasure channel with each entry being erased independently with probability $\epsilon$.
 The erasures characterize the missing entries in a rating matrix, while the BSC characterizes the noisy behaviour of the users. The entries of the  observed matrix $\mathbf Y$ are from $\{0,1,*\}$, where $*$ denotes an erased entry.

We consider the case of binary entries and uniform cluster size is for simplicity, and like in \cite{Aditya2}, these can  be relaxed. For more detailed motivation of this model, we refer to \cite{Aditya1},\cite{Aditya2}.

\subsection{A Local Popularity Algorithm}
\label{subsec:local}
Without loss of generality  suppose the first row belongs to $A_1$. Upon observing $\mathbf Y$, we want to recommend an item (a column) to the user 1.  In this paper we study a particular ``local'' algorithm, which only uses pairwise row correlations.
Let the number of commonly sampled entries between two rows (similarity)
$s_{ij}:=\sum_{k=1}^n\mathbf 1_{\{\mathbf Y(i,k)\neq *\}}\cdot \mathbf 1_{\{\mathbf Y(j,k)\neq *\}}\cdot\mathbf  1_{\{\mathbf Y(i,k)=\mathbf Y(j,k)\}},$
where $\mathbf 1_{\{.\}}$ denotes the indicator function.
We use the following local algorithm ($\texttt{local\_algo}(T)$) to recommend an item $j_0$ to user 1.


\vspace{0.03in}
\noindent
\fbox{\parbox{3.4in}{ 
\texttt{local\_algo($T$):}
\begin{enumerate}
\item ({\bf Select the top $T$ nearest rows}) Compute $s_{1i}$, for 
      $i=1,2,\dots,n$. Select the top $T$ rows with the highest values of similarity, where $T$ is a parameter whose choice is discussed later. 
      
\item ({\bf Pick the most popular column}) Among the columns $j$ such that $\mathbf Y(1,j)=*$, select the column having maximum number of 1's among the top $T$ neighbors. Break ties randomly. 
\end{enumerate}
}}
\vspace{0.01in}\\
Suppose we represent each row by a vertex in a graph with an edge between vertex $i$ and $j$ iff $s_{ij}>0$. Then to recommend an item to user 1, the above algorithm depends only on the rows neighboring to user 1, and chooses the most popular item among the top few neighbors. Hence we use the adjective ``local popularity''.
We study the probability of error for this algorithm, denoted as  $P_e[\texttt{local\_algo}(T)]:=Pr[\mathbf X(1,j_0)=0]$.

\section{Main Result}
\label{sec:main}

From the results in \cite{Aditya2}, it follows that for $k > c_1n^\alpha \log n$, $\alpha \in [0,1/2)$, with high probability we can recover the entire matrix $\mathbf X$ using a ``local'' algorithm, and hence the BER also approaches zero. In the following theorem, we establish a stronger result for $\texttt{local\_algo}(T)$.
\begin{theorem}
  \label{thm:01}
Suppose $\alpha\in [0,1/2)$ and $c>0$.  Assume that the erasure probability $\epsilon =1-\frac{c}{n^\alpha}$,  the BSC error probability  $p\in [0,1/2)$, and $r$ goes to infinity with $n$. 
\begin{itemize}
\item (\textbf {Large cluster size}) If there exists a sequence $\gamma_n\ge 0$ such that $\gamma_n \rightarrow 0$ and $k \ge n^{\alpha-\gamma_n}$, then $P_e[\texttt{local\_algo}(k)]\rightarrow 0$ as $n\rightarrow \infty$.
\item (\textbf {Small cluster size}) If there is a constant $\gamma > 0$ such that $k \le n^{\alpha -\gamma}$,  then 
$$\lim \inf_{n\rightarrow \infty} P_e[\texttt{local\_algo}(k)]  \ge \frac{p^{\left\lfloor \frac{1}{\gamma}\right\rfloor}}{p^{\left\lfloor \frac{1}{\gamma}\right\rfloor} + (1-p)^{\left\lfloor \frac{1}{\gamma}\right\rfloor}}.$$
\end{itemize}
\end{theorem}

In Theorem \ref{thm:01} we restrict ourselves to $\alpha \in [0, 1/2)$. For $\alpha < 1/2$, as we show in Section \ref{sec:proof_main}, all the rows picked by Step 1 of the algorithm are from $A_1$ (``good'') with high probability. But, for $\alpha >1/2$,  most of the rows picked are from outside $A_1$ (``bad''), and hence the algorithm breaks down. Due to lack of space, the results for $\alpha > 1/2$ will be presented in subsequent publications.
In the rest of this paper, we present a proof of Theorem \ref{thm:01}. 


\section{Proof of Theorem \ref{thm:01}}
\label{sec:proof_main}
In this section we present the proof of Theorem \ref{thm:01}.  To begin with, we introduce some notation.

\noindent
{\bf Notation:}
By $X\sim B(n,p)$ we mean that a random variable $X$ is binomially distributed with parameters $n$ and $p$. 
For a real valued function $f(n)$, by $\Omega(f(n)), \Theta(f(n))$ and $o(f(n))$ we represent the standard asymptotic order notation (see for example\cite[p. 433]{motwani95}).
We say that $f(n) \doteq g(n)$ if $\lim_{n\rightarrow \infty} \frac{f(n)}{g(n)}=1$. For a matrix $\mathbf X$, $\mathbf X(:,j)$ denotes the $j$th column of $\mathbf X$. For a vector $\bar y\in \{0,1,*\}^n$, $|\bar y|_0$, $|\bar y|_1$ and $|\bar y|$ represent number of 0's, number of 1's and the total number of 0's and 1's respectively. For a sequence of events $\{E_n\}$, if $P[E_n]\rightarrow 1$ with $n$, then we say that  $E_n$ occurs w.h.p..

\noindent
{\bf Analysis of Step 1 of the algorithm:}
We show that  w.h.p. the top $k$ rows are all from $A_1$.
We observe that for $i\in A_1\backslash\{1\}$, $s_{1i} \sim B(n, p_g)$ with  $p_g:=(1-\epsilon)^2[(1-p)^2+p^2]$. For $i\not \in A_1$ we observe that $s_{1i}$ is a mixture of binomials with $\mathbb E[s_{1i}]=np_b$ for $p_b:=\frac{(1-\epsilon)^2}{2}< p_g$. 
We omit the proofs of the following two lemmas, which are consequences of the Chernoff bound \cite[Theorem 1.1]{Dubhashi1} together with a union bound.
\begin{lemma}[{ \bf Overlap with ``good'' rows}]
  \label{lemma:01}
For $\delta\in (0,1)$, we have 
$$Pr\big[\min_{i\in A_1} s_{1i} \le np_g(1-\delta)\big] \le ke^{-np_g\delta^2/3}=:p_1.$$
\end{lemma}

\begin{lemma}[{\bf Overlap with ``bad'' rows}]
  \label{lemma:02}
For $\delta\in (0,1)$, we have
  $$Pr\big[\max_{i\not\in A_1}s_{1i} \ge np_b(1+\delta)^2\big] \le(n-k)e^{-\frac{np_b \delta^2}{3}}+2re^{-\frac{r\delta^2}{6}}=:p_2.$$
\end{lemma}

Since $p_g> p_b$, we can choose a small enough constant $\delta_0$ such that $np_g(1-\delta_0) > np_b(1+\delta_0)^2  $.
Let $E_1$ denote the event that there is an error in Step 1 of the algorithm, i.e., we choose some rows from outside $A_1$ in the top $k$ users. 
Using Lemma \ref{lemma:01} and Lemma  \ref{lemma:02} we obtain
\begin{align}
\label{eq:1}    Pr[E_1] & \le Pr\left[ \min_{i\in A_1} s_{1i} \le \max_{i\not\in A_1} s_{1i}\right]
 \le p_1 +p_2 \buildrel (a)\over = o(1).
  \end{align}
  Here (a)   follows since $np_g=\Theta(np_b)=\Theta(n^{1-2\alpha})$,  and $r$ increases to infinity with $n$. This implies that w.h.p. Step 1 of $\texttt{local\_algo}$ does not contribute to the error.

\noindent{\bf Analysis of Step 2 of the algorithm:}  We assume that Step 1 picks all the $k$ ``good'' neighbors. (i.e., we condition on the event $E_1^c$.) 

\noindent {\bf Large cluster size:}
Suppose $k\ge n^{\alpha -\gamma_n}$ for $\gamma_n=o(1)$. Let $j_{max}$ denote the most popular column chosen by \texttt{local\_algo}($k$), and suppose $\mathbf X_k$ and $\mathbf Y_k$ denotes the matrices $\mathbf X$ and $\mathbf Y$ respectively, restricted to the top $k$ rows.  Since we have conditioned on $E_1^c$, we observe that for a column $j$ such that $\mathbf X(1,j)=1$, we  have $|\mathbf Y_k(:,j)|_1\sim B(k,(1-\epsilon)(1-p))$. Define $\mu_Y:=\mathbb E[|\mathbf Y_k(:,j)|_1 ]$ and $\sigma_Y^2:=Var(|\mathbf Y_k(:,j)|_1 )$ to obtain the following two lemmas.
\begin{lemma}[{\bf Many 1's in the most popular column}]
 \label{lemma:101}
 For different values of $k$, we have the following lower bounds on $|\mathbf Y_k(:,j_{max})|_1$.
 \begin{enumerate}
 \item If $k=n^{\alpha-\gamma_n}$ such that $\gamma_n\ge 0$ and $\gamma_n\rightarrow 0$, then w.h.p. $|\mathbf Y_k(:,j_{max})|_1\ge \min\{\sqrt{\log n}, \frac{1}{2\gamma_n}\} =:t_1(n)$.
 \item If $k=n^\alpha g_n$ for $g_n \ge 1$, then w.h.p. $|\mathbf Y_k(:,j_{max})|_1\ge \max\{\mu_Y+\min\{\sigma_Y^{1/4},\sqrt{\log n}\} \sigma_Y,\sqrt{\log n}\}=:t_2(n)$.
 \end{enumerate}
\end{lemma}
\begin{proof}
The proof is given in Section \ref{proof:lemma:101}
\end{proof}

\begin{lemma}[{\bf 1's form majority in the most popular column}]
  \label{lemma:102}
Let $j_{max}$ be the most popular column. Then w.h.p.
$|\mathbf Y_k(:,j_{max})|_1 -|\mathbf Y_k(:,j_{max})|_0$ increases to $\infty$ with $n$.
\end{lemma}
   \begin{proof}
The proof is given in Section \ref{proof:lemma:102}
\end{proof}

 Now we use Lemma \ref{lemma:101} and Lemma \ref{lemma:102} to prove that the local algorithm makes vanishingly small probability of error. We define $t(k,n):=t_1(n)$ if $k=n^{\alpha-\gamma_n}$ for $\gamma_n\rightarrow 0$, and $t(k,n):=t_2(n)$ if $k=n^\alpha g_n$ for $g_n \ge 1$ (here $t_1(n)$ and $(t_2(n)$ are as defined in Lemma \ref{lemma:101}). Suppose  
$$M:=\{\bar y\in \{0,1,*\}^k: (|\bar y|_1-|\bar y|_0)\rightarrow \infty, \text{ and }|\bar y|_1 \ge t(k,n)\}.$$
 We also observe that for a column $j$, 
\begin{align}
  \label{eq:103}
  \mathbf X_k(:,j) \longrightarrow \mathbf Y_k(:,j)\longrightarrow \{j_{max}=j\},
\end{align}
i.e., the random variables $\{\mathbf X_k(:,j), \mathbf Y_k(:,j), \{j_{max}=j\}\}$ form a Markov chain.
We are interested in finding the overall probability of error. In the following, by $p_{k,j}(\bar y)$ we mean $Pr[\mathbf Y_k(:,j)=\bar y|j_{max}=j, E_1^c]$. Then we have 
\begin{align}
\nonumber  &  P_e[\texttt{local\_algo}(k)] = Pr[\mathbf X(1,j)=0|j_{max}=j]\\
\nonumber   \buildrel(a)\over = & Pr[\mathbf X(1,j)=0|j_{max}=j, E_1^c] +o(1)\\
\nonumber   \buildrel (b)\over = & \sum_{\substack{\bar y \in \{0,1,*\}^k\\ \bar y \in M}} \hspace{-0.2in}Pr[\mathbf X(1,j)=0, \mathbf Y_k(:,j)=\bar y|j_{max}=j, E_1^c] + o(1)\\
\nonumber   \buildrel (c)\over = &\sum_{\substack{\bar y \in \{0,1,*\}^k\\\bar y \in M}} Pr[\mathbf X(1,j)=0\big|\mathbf Y_k(:,j)=\bar y, E_1^c]\cdot p_{k,j}(\bar y) + o(1)\\
\nonumber   \buildrel (d)\over = &\sum_{\substack{\bar y \in \{0,1,*\}^k\\\bar y \in M}} \frac{Pr[\mathbf Y_k(:,j)=\bar y\big|\mathbf X(1,j)=0, E_1^c]}{2Pr[\mathbf Y_k(:,j)=\bar y|E_1^c]} p_{k,j}(\bar y) + o(1)\\
\nonumber   =&\sum_{\substack{\bar y \in \{0,1,*\}^k\\\bar y \in M}} \frac{p^{|\bar y|_1}(1-p)^{|\bar y|_0}}{p^{|\bar y|_1}(1-p)^{|\bar y|_0} + p^{|\bar y|_0}(1-p)^{|\bar y|_1}} p_{k,j}(\bar y) + o(1)
\end{align}
\begin{align}
\nonumber  \le &\max_{\substack{\bar y \in \{0,1,*\}^k\\\bar y \in M}} \frac{p^{|\bar y|_1-|\bar y|_0}}{p^{|\bar y|_1-|\bar y|_0} + (1-p)^{|\bar y|_1-|\bar y|_0}}+ o(1)\\
\label{eq:2}  & \buildrel (e) \over = o(1),
\end{align}
where (a) follows from \eqref{eq:1}, (b) is true because of Lemma \ref{lemma:101} and Lemma \ref{lemma:102}, (c) is due to the Markov property \eqref{eq:103} and the notation of $p_{k,j}(\bar y)$,  (d) is the Bayes' expansion, and (e) is true since for $\bar y\in M$, $|\bar y|_1-|\bar y|_0$  goes to infinity with  $n$, and the fact that $\frac{p^x}{p^x+(1-p)^x}=o(x)$ for $p<1/2$. This proves the first part of Theorem \ref{thm:01}.

\noindent
{\bf Small cluster size:}
Now suppose  $k\le n^{\alpha -\gamma}$ for a constant $\gamma >0$. We show that in this case the most popular column has a finite number of  unerased entries. This allows us to find a lower bound on the probability of error.
\begin{lemma}[{\bf Finite  number of unerased entries}]
  \label{lemma:103}
W.h.p. $$\max_j |\mathbf Y_k(:,j)| \le \lfloor 1/\gamma\rfloor.$$ 
\end{lemma}
The proof is based on bounding the tail of $\mathbf Y_k(:,j)$ and is not given here due to space restrictions.
Suppose $$I:=\{\bar y \in \{0,1,*\}^k : |\bar y|\le \lfloor 1/\gamma\rfloor \}.$$
 We want to find a lower bound for the total probability of error. By following the steps as in \eqref{eq:2} and replacing the event $M$ by the event $I$ (this replacement is justified due to Lemma \ref{lemma:103}), we have 
\begin{align*}
&  P_e[\texttt{local\_algo}(k)]  \\
 = &\sum_{\substack{\bar y \in \{0,1,*\}^k\\\bar y \in I}} \frac{p^{|\bar y|_1-|\bar y|_0}}{p^{|\bar y|_1-|\bar y|_0} + (1-p)^{|\bar y|_1-|\bar y|_0}}p_{k,j}(\bar y)+ o(1)\\
 \ge &\min_{\substack{\bar y \in \{0,1,*\}^k\\\bar y \in I}} \frac{p^{|\bar y|_1-|\bar y|_0}}{p^{|\bar y|_1-|\bar y|_0} + (1-p)^{|\bar y|_1-|\bar y|_0}}+ o(1)\\
 \buildrel (a) \over\ge &\frac{p^{\lfloor \frac{1}{\gamma}\rfloor}}{p^{\lfloor \frac{1}{\gamma}\rfloor} + (1-p)^{\lfloor \frac{1}{\gamma}\rfloor}}+o(1)
\end{align*}
where (a) is trues since  $|\bar y|_1 - |\bar y|_0 \le |\bar y| \le \lfloor 1/\gamma\rfloor$ for $\bar y \in I$, and for $x\in \mathbb R$, $\frac{p^{x}}{p^{x}+ (1-p)^{x}}$ is a decreasing function of $x$ for $p<1/2$. Taking $\lim \inf$ to both the  sides  proves the claim.

\section{Proofs of lemmas}
\label{sec:proof_lemma}
To prove Lemma \ref{lemma:101} and Lemma \ref{lemma:102}, we  need the following theorem.
 Suppose $Q(t)$ denotes the upper tail of a standard normal distribution, i.e., $Q(t):=\frac{1}{\sqrt{2\pi}} \int_t^\infty e^{-t^2/2} dt$.
\begin{theorem}[{\bf Moderate deviations for binomial distribution}]
\label{thm:large}
Suppose $X_n\sim B(n,p_n)$. If $t_n\rightarrow \infty$ in such a way that $t_n^6=o\left(Var(X_n)\right)=o(np_n(1-p_n))$, then
$$Pr\big[X_n > np_n + t_n \sqrt{np_n(1-p_n)}\big] \doteq Q(t_n).$$
\end{theorem}
The above theorem is an adaptation of a theorem about moderate deviations of binomials when $p_n$ is a constant \cite[p. 193]{Feller1}.  The proof is very similar to the one presented in \cite{Feller1} for the constant probability case, and is omitted here due to lack of space.

\subsection{Proof of Lemma \ref{lemma:101}}
\label{proof:lemma:101}
1) 
Recall that we have conditioned on the event that  all the rows in the top $k$ neighbors chosen by $\texttt{local\_algo}(k)$ are ``good''. Suppose $k=n^{\alpha-\gamma_n}$. Let $S$ be the set of columns $j$ such that $\mathbf X(1,j)=1$. Thus $|S|\sim B(n,1/2)$ and due to Chernoff bound we have w.h.p.  $|S|\ge n/3$. For a column $j\in S$ we see that $|\mathbf Y_k(:,j)|_1\sim B(k,(1-\epsilon)(1-p))$, and they are independent for different values of $j$. Thus for $j\in S$,
    \begin{align}
\nonumber     & Pr\big[|\mathbf Y_k(:,j)|_1\ge t\big] \ge Pr\big[|\mathbf Y_k(:,j)|_1=t\big]\\
\nonumber \buildrel (a)\over \ge &{k \choose t} ((1-\epsilon)(1-p))^t \epsilon^{k-t}\\
\label{eq:222}       \buildrel(b)\over  \ge &\left(\frac{k}{t}\right)^t \left(\frac{c(1-p)}{n^\alpha}\right)^t e^{-2\ln (2) c/n^{\gamma_n}}, \text{ for large $n$}\\
\nonumber      \buildrel (c) \over \ge &\left(\frac{c(1-p)}{tn^{\gamma_n}}\right)^t e^{-2\ln (2) c}.
    \end{align}
    where (a) is true since $1-(1-\epsilon)(1-p) \ge \epsilon$, (b) follows since $\epsilon=1-c/n^\alpha$, $1-x \ge e^{-2\ln (2) x}$ for $x\in [0,1/2]$, and ${k \choose t} \ge \left(\frac{k}{t}\right)^t$ (see \cite[p. 434]{motwani95}), and (c) is true because $\gamma_n \ge 0$.
    Since w.h.p. $|S|\ge n/3$,  we now have 
    \begin{align}
      \nonumber     &  Pr[|\mathbf Y_k(:,j_{max})|_1 < t] \\
      \nonumber    \le & Pr\left[\max_{j\in S} |\mathbf Y_k(:,j)|_1 < t \big| |S|\ge n/3 \right]+o(1)\\
      \nonumber \le & \left( 1- \left(\frac{c(1-p)}{tn^{\gamma_n}}\right)^t e^{-2\ln (2) c}\right)^{n/3}+o(1)\\
      \label{eq:201} \le & e^{-\frac{n}{3}\left(\frac{c(1-p)}{tn^{\gamma_n}}\right)^t e^{-2\ln (2) c} }+o(1)
    \end{align}
Suppose we put $t= t_0:=\min\{\sqrt{\log n}, \frac{1}{2\gamma_n}\}$. Then
$$\left(\frac{tn^{\gamma_n}}{c(1-p)}\right)^t = \frac{ n^{\gamma_n t} t^t}{(c(1-p))^t} \buildrel (a) \over \le \sqrt n  \left(\frac{\sqrt{\log n}}{c(1-p)}\right)^{\sqrt{\log n}}\hspace{-0.1in}=o(n),$$
where (a) follows since $\gamma_nt \le 1/2$ and $t\le \sqrt{\log n}$. Thus from \eqref{eq:201} we obtain
\begin{align*}
 & Pr[|\mathbf Y_k(:,j_{max})|_1 < t_0] 
\le e^{-\frac{1}{o(1)}}+o(1)=o(1).
\end{align*}
This proves the first part of the lemma.

2) Recall that we have assumed $k=n^\alpha g_n$ for $g_n \ge 1$. By following a very similar analysis as in the first part, we see that w.h.p. $|\mathbf Y_k(:,j_{max})|_1 \ge \sqrt{\log n}$. In particular for $g_n=1$ (or equivalently for $k=n^\alpha$),  \eqref{eq:222} becomes 
\begin{align}
\nonumber   Pr\big[|\mathbf Y_k(:,j)|_1\ge t\big] 
 \ge & \left(\frac{k}{t}\right)^t \left(\frac{c(1-p)}{n^\alpha}\right)^t e^{-2\ln (2) c}\\
\label{eq:555}  & =\left(\frac{c(1-p)}{t}\right)^t e^{-2\ln (2) c}.
    \end{align}
Observe that  for two random variables $X$ and $Y$ such that $X\sim B(n_1,p)$ and $Y\sim B(n_2,p)$ with $n_1 \ge n_2$, we have $Pr[X\ge t] \ge Pr[Y\ge t]$. Thus using \eqref{eq:555} we have 
\begin{align*}
Pr\left[|\mathbf Y_k(:,j)|_1\ge t\big|g_n\ge 1\right]  &\ge Pr\left[|\mathbf Y_k(:,j)|_1\ge t\big|g_n=1\right] \\
&\ge \left(\frac{c(1-p)}{t}\right)^t e^{-2\ln (2) c}.
\end{align*}
Hence for $t=\sqrt{\log n}$, \eqref{eq:201} has the following counterpart,
\begin{align*}
          Pr[|\mathbf Y_k(:,j_{max})|_1 < t] 
\le &  e^{-\frac{n}{3}\left(\frac{c(1-p)}{t}\right)^t e^{-2\ln (2) c} }+o(1)\\
&=e^{-\frac{1}{o(1)}}+o(1)=o(1).
    \end{align*}

But in Lemma \ref{lemma:102} we need better bounds for $g_n \rightarrow \infty$, and we consider this case now. Recall that for $j\in S$, $\mu_Y=\mathbb E[|\mathbf Y_k(:,j)|_1 ]=c(1-p)g_n$ and $\sigma_Y^2=Var(|\mathbf Y_k(:,j)|_1 )=g_nc(1-p)(1-(1-\epsilon)(1-p))$. We define $t_n:=\min\{\sigma_Y^{1/4}, \sqrt{\log n}\}$. Since $\sigma_Y\rightarrow \infty$, we have $t_n^6=o(\sigma_Y^2)$, and then Theorem \ref{thm:large} implies that for a column $j\in S$,
\begin{align*}
  Pr[|\mathbf Y_k(:,j)|_1 > \mu_Y +t_n \sigma_Y] \doteq& Q(t_n)
  \buildrel (a)\over \doteq \frac{1}{\sqrt{2\pi}t_n}e^{-t_n^2/2}\\ 
   \ge &\frac{1}{2}\frac{1}{\sqrt{2\pi}t_n}e^{-t_n^2/2}, \text{ for large $n$}\\
  &\buildrel (b)\over = \Omega\left(\frac{1}{\sqrt{n \log n}}\right).
\end{align*}
where (a) is true because $Q(t)\doteq \frac{1}{\sqrt{2\pi}t}e^{-t^2/2}$ \cite[Lemma 1.2]{Feller1}, and (b) is true since $t_n \le \sqrt{\log n}$. Since w.h.p. $|S|\ge n/3$, we have
\begin{align}
\nonumber & Pr[|\mathbf Y_k(:,j_{max})|_1  \le \mu_Y +t_n \sigma_Y]\\
\nonumber \le & Pr\left[\max_{j\in S} |\mathbf Y_k(:,j)|_1 \le \mu_Y +t_n \sigma_Y \big| |S|\ge n/3\right]+o(1)\\
\nonumber  \le &\left(1-\Omega\left(\frac{1}{\sqrt{n \log n}}\right)\right)^{n/3}+o(1)=o(1).
\end{align}
Thus w.h.p. $|\mathbf Y_k(:,j_{max})|_1  \ge \mu_Y +t_n \sigma_Y$, if $g_n \rightarrow \infty$. We have already observed that w.h.p. $|\mathbf Y_k(:,j_{max})|_1 \ge \sqrt{\log n}$. Thus the lemma is implied.

\subsection{Proof of Lemma \ref{lemma:102}}
\label{proof:lemma:102}
Lemma \ref{lemma:101} gives us a lower bound for $|\mathbf Y_k(:,j_{max})|_1$ that holds w.h.p.. Next we  find an upper bound for  $|\mathbf Y_k(:,j_{max})|_0$ to prove Lemma \ref{lemma:102}. 

 First we condition on the event that $\mathbf X(1,j_{max})=1$. We observe that
$$|\mathbf Y_k(:,j)|_0 \longrightarrow |\mathbf Y_k(:,j)|_1 \longrightarrow \{j_{max}=j\}.$$
Then conditioned on the value of $|\mathbf Y_k(:,j_{max})|_1=t$, the distribution of $|{\mathbf Y_k(:,j_{max})}|_0$ does not depend on the fact that $j_{max}$ is the most popular column chosen by the algorithm, and hence $|{\mathbf Y_k(:,j_{max})}|_0 \sim B\left(k-t, p_0\right)$, where $p_0:=\frac{p(1-\epsilon)}{p(1-\epsilon)+\epsilon}$.  This is because for a given column $j$ of $\mathbf Y_k$, upon observing that there are exactly $t$ 1's, the other $k-t$ entries are i.i.d. with probability of 0 being $p_0$.

1) Suppose $k=n^{\alpha - \gamma_n}$ such that $\gamma_n\rightarrow 0$. We define $b(k,p,i):={ k\choose i} p^i (1-p)^{n-i}$ to be the $i$th binomial term, and observe that $b(k,p,i)\le \left(kpe/i\right)^i$, since ${k\choose i}\le (ke/i)^i$ (see \cite[p. 434]{motwani95}). We see that 
\begin{align*}
&  Pr\left[|{\mathbf Y_k(:,j_{max})}|_0  \ge \frac{\sqrt{\log n}}{2}\right] =\sum_{i=\frac{\sqrt{\log n}}{2}}^{k-t} b(k-t, p_0,i)\\
 =& \sum_{i=\frac{\sqrt{\log n}}{2}}^{2\log n}b(k-t, p_0,i) + \sum_{i=2\log n+1}^{k-t} b(k-t, p_0,i)\\
\buildrel (a)\over \le & 2\log n \cdot b\left(k-t, p_0, \frac{\sqrt{\log n}}{2}\right) + k\cdot b(k-t, p_0, 2\log n+1)\\
\buildrel (b) \over \le &2\log n \left(\frac{(k-t)p_0e}{\sqrt{\log n}/2}\right)^{\frac{\sqrt{\log n}}{2}} \hspace{-0.1in}+\hspace{-0.05in} (k-t) \left(\frac{(k-t)p_0e}{2\log n+1}\right)^{2\log n+1}\\ 
\buildrel (c)\over \le &2 \log n \left(\frac{2c'}{n^{\gamma_n}\sqrt{\log n}}\right)^{\frac{\sqrt{\log n}}{2}} \hspace{-0.1in}  + k \left(\frac{c'}{n^{\gamma_n} (2\log n+1)}\right)^{2\log n+1}\\
& = o(1).
\end{align*}
where (a) is true since $b(k,p,i)$ is a decreasing function of $i$ for $i$ more than $kp$ and we have $(k-t)p_0=o(1)$, (b) is due to the fact that $b(k,p,i)\le (kpe/i)^i$ , and (c) follows by observing that $kp_0e\le \left(\frac{c'}{n^{\gamma_n}}\right)$ for a constant $c'>0$. Thus 
w.h.p. we have $|{\mathbf Y_k(:,j_{max})}|_0  < \frac{\sqrt{\log n}}{2}$.

Now suppose $\gamma_n > \frac{1}{2\sqrt{\log n}}$. Then we see that 
\begin{align*}
  &Pr\left[|{\mathbf Y_k(:,j_{max})}|_0  \ge \frac{1}{4\gamma_n}\right]  =\sum_{i=\frac{1}{4\gamma_n}}^{k-t} b(k-t, p_0,i) \\
  \le & \sum_{i=\frac{1}{4\gamma_n}}^\infty b(k-t, p_0,i) \buildrel (a)\over \le \sum_{i=\frac{1}{4\gamma_n}}^\infty ((k-t) p_0e/i)^i\\  
\buildrel (b)\over \le & \sum_{i=\frac{1}{4\gamma_n}}^\infty \left( \frac{4c' \gamma_n}{n^{\gamma_n}}\right)^i \buildrel (c)\over = \Theta\left(\left( \frac{4c' \gamma_n}{n^{\gamma_n}}\right)^{1/4\gamma_n}\right)\\
 \buildrel (d) \over < & \Theta\left(\frac{(4c' \gamma_n)^{\sqrt{\log n}/2}}{n^{1/4}}\right)=  o(1),
\end{align*}
where (a) is true since $b(k,p,i)\le (kpe/i)^i$, (b) follows because $kp_0e\le \left(\frac{c'}{n^{\gamma_n}}\right)$ for a constant $c'$, (c) is true by observing that for $x=o(1)$, we have $\sum_{i=m}^\infty x^i = \Theta(x^m)$, and (d) follows since $\frac{1}{4\gamma_n} < \frac{\sqrt{\log n}}{2}$ whenever $\gamma_n > \frac{1}{2\sqrt{\log n}}$. 

Thus we have proved that w.h.p. $|{\mathbf Y_k(:,j_{max})}|_0  < \min \{ \frac{\sqrt{\log n}}{2}, \frac{1}{4\gamma_n}\}$. This together with the observation in Lemma \ref{lemma:101} that w.h.p. $|{\mathbf Y_k(:,j_{max})}|_1  \ge \min \{ \sqrt{\log n}, \frac{1}{2\gamma_n}\}$, proves that w.h.p. $|{\mathbf Y_k(:,j_{max})}|_1 -|{\mathbf Y_k(:,j_{max})}|_0$ increases to $\infty$ with $n$.

2) Now we consider the other case of $k=n^\alpha g_n$ for $g_n\ge 1$. If $g_n$ is upper bounded by a constant, then arguments very similar to those used in the first part tell us that w.h.p. $|{\mathbf Y_k(:,j_{max})}|_0 <\sqrt{\log n}/2$.  

In the remaining part of the proof, we assume that $g_n\rightarrow \infty$. Recall that for a column $j$ such that $\mathbf X(1,j)=1$, we have  $\mu_Y=\mathbb E[|\mathbf Y_k(:,j)|_1 ]=k(1-\epsilon)(1-p)$ and $\sigma_Y^2=Var(|\mathbf Y_k(:,j)|_1 )=k(1-p)(1-\epsilon)(1-(1-\epsilon)(1-p))$. Conditioned on the value of  $|\mathbf Y_k(:,j_{max})|_1=t$, suppose $\mu_{\bar Y}$ and $\sigma_{\bar Y}^2$ denote the conditional mean and variance of $|\mathbf Y_k(:,j_{max})|_0$.  We observe that for $t\ge  \mu_Y$ and large enough $n$,
$$\mu_{\bar Y}=(k-t)p_0\le \mu_Y,\text{ and } \sigma_{\bar Y}^2 =(k-t)p_0(1-p_0)\le 2 \sigma_Y^2.$$
Suppose $t_n:=\min\{\sigma_Y^{1/4},\sqrt{\log n}\} $. Since $\sigma_Y\rightarrow \infty$, we have $t_n^6=o(\sigma_Y^2)$, and since w.h.p. $y_1:=|\mathbf Y_k(:,j_{max})|_1\ge \mu_Y$ (see Lemma \ref{lemma:101}), using Theorem \ref{thm:large} we obtain
\begin{align*}
&  Pr\left[|{\mathbf Y_k(:,j_{max})}|_0  > \mu_{Y}+ \frac{t_n}{2} \sigma_{Y}\right] \\
 \le &Pr\left[|{\mathbf Y_k(:,j_{max})}|_0  > \mu_{\bar Y}+ \frac{1}{2\sqrt 2} t_n \sigma_{\bar Y}\big| y_1\ge \mu_{Y}\right] +o(1)\\
&  \doteq Q\left(\frac{t_n}{2\sqrt 2 }\right) = o(1). 
\end{align*}
Thus w.h.p.  $|{\mathbf Y_k(:,j_{max})}|_0 \le \max\{\sqrt{\log n}/2, \mu_Y+\frac{t_n}{2}\sigma_Y\}$. This together with the observation made in Lemma \ref{lemma:101} that w.h.p. $|{\mathbf Y_k(:,j_{max})}|_1 \ge\max\{\sqrt{\log n}, \mu_Y+t_n\sigma_Y\}$, proves that w.h.p. $|{\mathbf Y_k(:,j_{max})}|_1 - |{\mathbf Y_k(:,j_{max})}|_0$ increases to $\infty$ with $n$.

\noindent
{\bf Remark: }
  In the above proof, we had conditioned on the event that $\mathbf X(1, j_{max})=1$. When we condition on  $\mathbf X(1,j_{max})=0$, we have $p_0=\frac{(1-p)(1-\epsilon)}{(1-p)(1-\epsilon)+\epsilon}$, and a  similar set of steps prove the claim.

\section{Conclusion}
\label{sec:con}
We have considered estimation of a binary random field obtained by permuting rows and columns of a block constant matrix, by observing a  sub-sampled and noisy version. It would be interesting to analyze the performance of ``local'' algorithms on a more general class of matrices obtained from realizations of  a ``smooth'' stochastic process. Further, non-uniform sampling models are also of interest.

\section*{Acknowledgment}

The work of Kishor Barman was supported by the Infosys fellowship and the  Microsoft Research India Travel Grants Program.. The work of Onkar Dabeer was supported by the XI plan funding.

\bibliographystyle{IEEEtran}
\bibliography{myBib}

\begin{thebibliography}{10}
\providecommand{\url}[1]{#1}
\csname url@samestyle\endcsname
\providecommand{\newblock}{\relax}
\providecommand{\bibinfo}[2]{#2}
\providecommand{\BIBentrySTDinterwordspacing}{\spaceskip=0pt\relax}
\providecommand{\BIBentryALTinterwordstretchfactor}{4}
\providecommand{\BIBentryALTinterwordspacing}{\spaceskip=\fontdimen2\font plus
\BIBentryALTinterwordstretchfactor\fontdimen3\font minus
  \fontdimen4\font\relax}
\providecommand{\BIBforeignlanguage}[2]{{%
\expandafter\ifx\csname l@#1\endcsname\relax
\typeout{** WARNING: IEEEtran.bst: No hyphenation pattern has been}%
\typeout{** loaded for the language `#1'. Using the pattern for}%
\typeout{** the default language instead.}%
\else
\language=\csname l@#1\endcsname
\fi
#2}}
\providecommand{\BIBdecl}{\relax}
\BIBdecl

\bibitem{netflixprize}
{Netflix prize}, \url{http://www.netflixprize.com/}.

\bibitem{Candes1}
E.~J. Cand{\`e}s and B.~Recht, ``Exact matrix completion via convex
  optimization,'' \emph{CoRR}, vol. abs/0805.4471, 2008.

\bibitem{Montanari1}
R.~H. Keshavan, S.~Oh, and A.~Montanari, ``Matrix completion from a few
  entries,'' \emph{CoRR}, vol. abs/0901.3150, 2009.

\bibitem{Bresler1}
K.~Lee and Y.~Bresler, ``Efficient and guaranteed rank minimization by atomic
  decomposition,'' \emph{CoRR}, vol. abs/0901.1898, 2009.

\bibitem{Aditya1}
S.~T. Aditya, O.~Dabeer, and B.~K. Dey, ``A channel coding perspective of
  recommendation systems,'' in \emph{IEEE International Symposium on
  Information Theory}, 2009, pp. 319--323.

\bibitem{Aditya2}
------, ``A channel coding perspective of collaborative filtering,''
  \emph{CoRR}, vol. abs/0908.2494, 2009.

\bibitem{Linden1}
G.~Linden, B.~Smith, and J.~York, ``Amazon.com recommendations: Item-to-item
  collaborative filtering,'' \emph{IEEE Internent Computing}, Jan./Feb. 2003.

\bibitem{MovieLens}
{MovieLens data}, \url{http://www.grouplens.org/node/73}.

\bibitem{motwani95}
R.~Motwani and P.~Raghavan, \emph{Randomized Algorithms}.\hskip 1em plus 0.5em
  minus 0.4em\relax Cambridge University Press, 1995.

\bibitem{Dubhashi1}
D.~Dubhashi and A.~Panconesi, \emph{Concentration of Measure for the Analysis
  of Randomised Algorithms}, 1st~ed.\hskip 1em plus 0.5em minus 0.4em\relax
  Cambridge University Press, 2009.

\bibitem{Feller1}
W.~Feller, \emph{An Introduction to Probability, theory and its applications
  [Volume I]}, 3rd~ed.\hskip 1em plus 0.5em minus 0.4em\relax Willey India,
  2008.

\end{thebibliography}

\end{document}